\documentclass[letterpaper,10pt,conference]{ieeeconf}  
\usepackage{mathtools}
\usepackage{bm}
\usepackage{amsfonts}
\usepackage[short,c2,nocomma]{optidef}
\usepackage{amsthm}
\makeatletter
\let\NAT@parse\undefined
\makeatother
\usepackage{hyperref}
\usepackage{cleveref}
\crefformat{figure}{Fig.~#2#1#3}
\Crefformat{figure}{Fig.~#2#1#3}
\usepackage{tikz}
\usepackage{pgfplots}
\usepgfplotslibrary{fillbetween,groupplots}
\usepackage{pgfplotstable}
\pgfplotsset{compat=1.18}
\usepackage{siunitx}
\usepackage[algo2e,ruled,linesnumbered]{algorithm2e}
\DeclareMathOperator*{\argmin}{arg \ min}
\DeclarePairedDelimiter{\norm}{\lVert}{\rVert}

\DeclarePairedDelimiterXPP{\prob}[1]{\mathbb{P}}[]{}{%
    
    #1
}
\DeclarePairedDelimiter{\stochasticargument}[]
\NewDocumentCommand{\expval}{e{_}}{%
    \mathbb{E}\IfValueT{#1}{_{#1}}\stochasticargument
}
\newcommand{\funcdef}[3]{#1 : #2 \rightarrow #3}
\newtheorem{assumption}{Assumption}
\newtheorem{definition}{Definition}
\newtheorem{theorem}{Theorem}

\definecolor{C0}{RGB}{90,91,159}  
\definecolor{C1}{RGB}{255,111,97}  
\definecolor{C2}{RGB}{0,148,115}  
\definecolor{C3}{RGB}{217,79,112}  
\definecolor{C4}{RGB}{123,196,196}  
\definecolor{C5}{RGB}{240,192,90}  
\definecolor{darkslategray}{RGB}{38,38,38}
\definecolor{lavender}{RGB}{234,234,242}
\pgfplotsset{
    every axis/.style={
        axis background/.style={fill=lavender},
        axis line style={white},
        tick align=outside,
        tick pos=left,
        enlarge x limits=0.05,
        enlarge y limits=0.05,
        tick scale binop=\times,
        x grid style={white},
        y grid style={white},
        xmajorgrids,
        ymajorgrids,
        xtick style={color=darkslategray},
        ytick style={color=darkslategray},
    },
    every axis title/.style={anchor=south, at={(0.5,1)}},
    every axis legend/.append style={
        cells={anchor=west},
        fill=lavender,
        draw=white,
    }
}
\def\axisdefaultwidth{\columnwidth} 
\def\axisdefaultheight{130pt}

\IEEEoverridecommandlockouts  
\newcommand{\mytitle}{Probabilistically safe and efficient model-based reinforcement learning}
\hypersetup{
    hidelinks,
    pdfauthor=Filippo Airaldi,
    pdfcreator=Filippo Airaldi,
    pdftitle=\mytitle,
}

\title{\LARGE \bf \mytitle}

\author{Filippo Airaldi, Bart De Schutter, and Azita Dabiri%
\thanks{This research is part of a project that has received funding from the European Research Council (ERC) under the European Union's Horizon 2020 research and innovation programme (Grant agreement No. 101018826 - CLariNet).}%
\thanks{The authors are with the Delft Center for Systems and Control, Delft University of Technology, 2628 CD Delft, The Netherlands {\tt\small \{f.airaldi,b.deschutter,a.dabiri\}@tudelft.nl}}%
}

\begin{document}

\maketitle
\thispagestyle{empty}
\pagestyle{empty}

\begin{abstract}
    This paper proposes tackling safety-critical stochastic Reinforcement Learning (RL) tasks with a sample-based, model-based approach. At the core of the method lies a Model Predictive Control (MPC) scheme that acts as function approximation, providing a model-based predictive control policy. To ensure safety, a probabilistic Control Barrier Function (CBF) is integrated into the MPC controller. To approximate the effects of stochasticies in the optimal control formulation and to fulfil the probabilistic CBF condition, a sample-based approach with guarantees is employed. Furthermore, to counterbalance the additional computational burden due to sampling, a learnable terminal cost formulation is included in the MPC objective. An RL algorithm is deployed to learn both the terminal cost and the CBF constraint. Results from a numerical experiment on a constrained LTI problem corroborate the effectiveness of the proposed methodology in reducing computation time while preserving control performance and safety.
\end{abstract}

\section{INTRODUCTION} \label{sections:introduction}

Reinforcement Learning (RL) has emerged as a successful methodology for solving complex optimal control problems, including when dealing with systems subject to uncertainty and stochastic disturbances \cite{sutton_2018_reinforcement}. However, employing RL in safety-critical scenarios remains in general challenging due to the inherent trial-and-error nature of the learning process and the difficulties in explicitly ensuring constraint satisfaction throughout training, even if probabilistically.  

Control Barrier Functions (CBFs) have gained significant traction as an effective tool for handling safety constraints in control problems \cite{ames_2019_control}. CBFs can enforce forward invariance of a safe set, thus guaranteeing safety conditions over the controlled trajectories, via an energy-based argumentation rather than relying on explicit set computations. Robust and stochastic extensions have also spun off, e.g., \cite{alan_2023_parametrized,clark_2019_control}, which account for uncertainties and/or disturbances affecting the system. At the same time, CBFs have been successfully integrated with various control architectures, including optimisation-based control schemes such as Model Predictive Control (MPC) \cite{zeng_2021_safety,nascimento_2024_probabilistically,sabouni_2024_reinforcement,wang_2024_stochastic}. 
While CBFs offer guarantees on safety, their usage introduces some challenges. One of these lies in properly calibrating the CBF parameters, particularly its class $\mathcal{K}$ function. Poor tuning can severely impact the feasibility of the control problem and its closed-loop performance. Selecting an appropriate class $\mathcal{K}$ function thus involves a non-trivial trade-off between conservativeness (safety) and control performance. In this regard, adaptive formulations have been proposed in the literature that, e.g., employ auxiliary constructions \cite{xiao_2022_adaptive} or leverage intelligent decision-making \cite{sabouni_2024_reinforcement,xiao_2023_barriernet} to adjust the barrier parameters autonomously.

Recently, MPC has been proposed as a promising function approximation strategy for RL algorithms, where the predictive controller acts both as policy provider and value function approximation for the underlying RL task \cite{gros_2020_datadriven,reiter_2025_synthesis}. In contrast to model-free approaches, this method often results in higher sample efficiency, better interpretability and certifiability, since the MPC controller can explicitly incorporate system dynamics and systematically handle constraints. Importantly, variants of the nominal MPC formulation can also address robust and/or stochastic control problems characterised by uncertainties and/or disturbances \cite{mesbah_2016_stochastic,schildbach_2014_scenario,campi_2019_scenario}. Despite its benefits, the application of stochastic MPC, particularly in its sample-based forms, is often computationally demanding. One common approach to mitigate this computational complexity is to shorten the MPC prediction horizon. However, doing so can adversely affect control performance and safety due to the induced myopia of the controller. This issue is commonly addressed by introducing a terminal cost approximation, crafted to appropriately approximate the true (generally unknown) cost-to-go beyond the shortened horizon \cite{karnchanachari_2020_practical,seel_2022_convex,abdufattokhov_2024_learning,he_2024_apprxoimate}. Nonetheless, similar to the class $\mathcal{K}$ functions in CBFs, manually selecting or designing an effective terminal cost approximation introduces another trade-off between computational complexity, safety, and control performance.

In this paper, we propose a novel approach that leverages MPC-based RL combined with probabilistic CBFs and terminal cost approximation to automatically learn from data a model-based policy that ensures probabilistic safety while maintaining computational efficiency. Our methodology integrates probabilistic CBF constraints into the MPC formulation to enforce safety despite stochastic disturbances with arbitrary probability. 
A sample-based approximation is introduced to render the optimisation control problem tractable. This is distinct from, e.g., MPPI \cite{williams_2016_aggressive}, where sampling is exploited to compute a Monte Carlo approximation of the optimal action sequence of the control problem.
To address the computational complexity introduced by the sample-based approach and additional CBF constraint, we employ a shortened MPC prediction horizon alongside a learnable terminal cost approximation, which is automatically tuned via RL. Furthermore, the class $\mathcal{K}$ function within the CBF is also learned from interaction data, eliminating manual tuning and enabling adaptivity. Algorithm \ref{algo:introduction:algorithm} summarises the proposed approach in a compact scheme.
\begin{algorithm2e}[t]
    \caption{Stochastic safe MPC-based RL. See \Cref{sections:methodology} for further details on each step below.}
    \label{algo:introduction:algorithm}
    \KwIn{Initial MPC parameters $\bm\theta^0$, violation prob. $\varepsilon$, full and short horizons $N$, $\overline{N}$, number of training episodes $n_{\max}$;}
    \KwOut{Learnt parametrisation $\bm\theta^{n_{\max}}$;}
    Select number of samples $M$ from $\varepsilon$, $N$, $\overline{N}$ empirically or via conservative bounds \eqref{eq:methodology:num-samples-nonconvex}, \eqref{eq:methodology:num-samples-convex}\;
    Create sampled-based MPC controller \eqref{eq:methodology:scmpc} with learnable parametric CBF constraints (to ensure safety) and terminal cost approximation (to compensate for $\overline{N}$)\;
    \For{$i = 0, \ldots, n_{\max}-1$}{
        Perform MPC closed-loop task with current parametrisation $\bm\theta^i$\;
        Update parametrisation to $\bm\theta^{i+1}$ via RL\;
    }
\end{algorithm2e}

The main contributions of this paper can thus be summarised as follows:
\begin{enumerate}
    \item We introduce a stochastic MPC formulation with integrated probabilistic CBF constraints, explicitly designed to handle stochasticity in safety-critical tasks.
    \item We provide a computationally efficient sample-based approximation of this formulation and propose to leverage RL to automatically learn both the terminal cost approximation and the CBF class $\mathcal{K}$ function.
    \item We provide probabilistic safety guarantees and illustrate the effectiveness and computational advantages of our proposed method on a numerical example.
\end{enumerate}

The remainder of this paper is structured as follows. In \Cref{sections:background}, we review relevant background on safe RL, MPC as function approximation, and terminal cost approximations. \Cref{sections:methodology} describes and analyses the proposed method. Simulation results validating our approach are presented in \Cref{sections:numerical-exp}, followed by conclusions in \Cref{sections:conclusions}.

\textit{Notation}: vector and matrix quantities are in bold. Inequalities on vectors are applied element-wise. Operation $\norm*{\bm{y}}_{\bm{A}}$ indicates $\sqrt{ \bm{y}^\top \bm{A} \bm{y} }$, and $\bm{y} \odot \bm{z}$ the Hadamard product between $\bm{y}$ and $\bm{z}$.

\section{BACKGROUND} \label{sections:background}

\subsection{Safe Reinforcement Learning}

Consider the discrete-time, possibly nonlinear, stochastic system
\begin{equation}  \label{eq:background:dynamics}
    \bm{s}_{t+1} = f\left( \bm{s}_t, \bm{a}_t, \bm\omega_t \right),
\end{equation}
where, at each time step $t \in \mathbb{N}$, $\bm{s}_t \in \mathcal{S} \subseteq \mathbb{R}^{n_s}$ denotes its state, $\bm{a}_t \in \mathcal{A} \subset \mathbb{R}^{n_a}$ the control action, and $\bm\omega_t \in \Omega \subseteq \mathbb{R}^{n_d}$ the disturbance affecting the system. Dynamics $\funcdef{f}{\mathcal{S} \times \mathcal{A} \times \Omega}{\mathcal{S}}$ are assumed to be known and Lipschitz continuous w.r.t. $\bm{s}_t$ and $\bm{a}_t$ with constant $L_f$ over the domain $\mathcal{S} \times \mathcal{A}$. 
\begin{assumption}[Uncertainty]  \label{assumptions:background:iid-samples}
    Sequences $\left\{ \bm\omega_\tau \right\}_{\tau = t}^{t+N} \sim \mathcal{W}$, for $N > 0$, are independent and identically distributed (i.i.d.) random variables with support $\Omega^N$. Further, a sufficient number of i.i.d. samples of these sequences can be drawn from $\mathcal{W}$ or is available (e.g., from historical data). 
\end{assumption}

For a deterministic policy $\funcdef{\pi_{\bm\theta}}{\mathcal{S}}{\mathcal{A}}$, parametrised in $\bm\theta \in \Theta \subseteq \mathbb{R}^{n_{\bm\theta}}$, we define its performance as\footnote{For simplicity, we address in this paper the finite-horizon undiscounted setting, but our results can be easily extended to the infinite-horizon discounted case.} 
\begin{equation}   \label{eq:background:safe-rl:performance}
    J\left( \pi_{\bm\theta} \right) \coloneqq \expval_{\chi_{\pi_{\bm\theta}}}*{
        \sum_{t=0}^T \ell \bigl( \bm{s}_t, \pi_{\bm\theta} \left( \bm{s}_t \right) \bigr)
    },
\end{equation}
where $T \in \mathbb{N}$ is the horizon, $\funcdef{\ell}{\mathcal{S} \times \mathcal{A}}{\mathbb{R}}$ is a stage cost function, and $\chi_{\pi_{\bm\theta}}$ the state distribution the policy induces. In safe RL, we are primarily interested in finding a policy that optimises the performance while providing safe trajectories with high probability, i.e., 
\begin{equation}  \label{eq:background:safe-rl}
    \pi_{\bm\theta}^\star \in \argmin_{ 
         \bm\theta \in \Theta 
    }{ \left\{ 
        J(\pi_{\bm\theta})
            \: : \:
        \prob*{ \bigcap_{t=0}^T \bm{s}_t \in \mathcal{C} } \geq 1 - \varepsilon
    \right\} },
\end{equation}
where $\mathcal{C} = \left\{ \bm{s} \in \mathcal{S} \ | \ h(\bm{s}) \geq 0 \right\}$ denotes the desired safe set, defined by $\funcdef{h}{\mathcal{S}}{\mathbb{R}}$, a Lipschitz continuous function in $\mathcal{S}$ with constant $L_h$, and $\varepsilon \in (0, 1)$ the confidence level for the joint chance constraint. Note that, due to the presence of stochastic disturbances, also the state becomes a random variable and generally cannot be constrained to satisfy $h$ with unit probability without further assumptions (e.g., boundedness of the support of $\bm\omega_t$). 

The familiar notions of state- and action-value functions \cite{sutton_2018_reinforcement} apply here as well: 
\begin{align}
    V_{\bm\theta}(\bm{s}_t) ={}& \expval_{\chi_{\pi_{\bm\theta}}}*{ 
        \sum_{\tau=t}^T \ell \bigl( \bm{s}_\tau, \pi_{\bm\theta} \left( \bm{s}_\tau \right) \bigr)
    },  \\
    Q_{\bm\theta}(\bm{s}_t, \bm{a}_t) ={}& 
        \ell(\bm{s}_t, \bm{a}_t)
        + \expval_{\bm\omega_t}[\big]{ 
            V_{\bm\theta}(\bm{s}_{t+1})
        }.
\end{align}

\subsection{MPC as Function Approximation in RL}

To parametrise the policy $\pi_{\bm\theta}$ and deploy an RL agent, (deep) neural networks are oftentimes the most common choice \cite{arulkumaran_2017_deep}. However, model-free approaches generally suffer from several drawbacks, as discussed in \Cref{sections:introduction}. In this work, a model-based solution to \eqref{eq:background:safe-rl}, which leverages MPC as the function approximation scheme, is instead pursued. 

Given the current state $\bm{s}_t$, consider the MPC scheme
\begin{mini!}
    { \scriptstyle \left\{ \bm{u}_k \right\}_{k=0}^{N-1} }{ 
        \lambda_{\bm\theta}(\bm{x}_0) 
        + \expval*{
            \sum_{k=0}^{N-1} \ell_{\bm\theta}(\bm{x}_k, \bm{u}_k)
            + V^\text{f}_{\bm\theta}(\bm{x}_N)
        }
        \label{eq:background:mpc:obj}
    }{ \label{eq:background:mpc} }{}
    \addConstraint{ 
        \bm{u}_k \in \mathcal{A}, 
        \hspace{64pt} k = 0,\ldots,N - 1, 
        \label{eq:background:mpc:constr:first} 
    }
    \addConstraint{ \bm{x}_0 = \bm{s}_t, }
    \addConstraint{ \bm{x}_{k+1} }{ = f(\bm{x}_k, \bm{u}_k, \bm\omega_k), \; k = 0,\ldots,N - 1,}
    \addConstraint{ 
        \left\{ \bm\omega_k \right\}_{k=0}^{N-1} \sim \mathcal{W}, 
        \nonumber
    }
    \addConstraint{ 
        \prob*{ \bigcap_{k=1}^N \bm{x}_k \in \mathcal{C} } 
        \geq 1 - \varepsilon, 
        \label{eq:background:mpc:constr:last}  
    }
\end{mini!}
where $N \geq 0$ is the prediction horizon, $\funcdef{\ell_{\bm\theta}}{\mathcal{S} \times \mathcal{A}}{\mathbb{R}}$ and $\funcdef{\lambda_{\bm\theta}, V_{\bm\theta}^\text{f}}{\mathcal{S}}{\mathbb{R}}$ the stage, initial, and final cost approximations respectively. This scheme serves as the approximation of the value function as
\begin{equation}  \label{eq:background:V}
    V_{\bm\theta}(\bm{s}_t) = \min_{ 
        \left\{ \bm{u}_k \right\}_{k=0}^{N-1} 
    }{ \left\{ \eqref{eq:background:mpc:obj} 
        \: : \:
        \eqref{eq:background:mpc:constr:first}-\eqref{eq:background:mpc:constr:last}
        \right\}
    }
\end{equation}
and it satisfies the Bellman equations so that 
\begin{gather}  
    Q_{\bm\theta}(\bm{s}_t, \bm{a}_t) = \min_{ 
        \left\{ \bm{u}_k \right\}_{k=0}^{N-1} 
    }{ \left\{ \eqref{eq:background:mpc:obj} 
        \: : \:
        \eqref{eq:background:mpc:constr:first}-\eqref{eq:background:mpc:constr:last},
        \, \bm{u}_0 = \bm{a}_t
        \right\}
    },  \label{eq:background:Q}  \\
    \pi_{\bm\theta}(\bm{s}_t) = \bm{u}_0^\star = \argmin_{ 
        \left\{ \bm{u}_k \right\}_{k=0}^{N-1} 
    }{ \left\{ \eqref{eq:background:mpc:obj} 
        \: : \:
        \eqref{eq:background:mpc:constr:first}-\eqref{eq:background:mpc:constr:last}
        \right\}
    }.
\end{gather}
It was first shown in \cite{gros_2020_datadriven} that the solution to an MPC optimisation problem can approximate the optimal value function. This is especially intuitive if the MPC horizon $N$ were to approach the task horizon $T$ and we would take $\ell_{\bm\theta} = \ell$ and $\lambda_{\bm\theta}, V_{\bm\theta}^\text{f} = 0$. However, in general, long prediction horizons and the stochastic arguments in \eqref{eq:background:mpc} massively hinder the tractability of the MPC problem. In \Cref{sections:methodology}, we present our approach to circumvent both issues in the context of safe RL. While in general, in addition to $V_{\bm\theta}^\text{f}$, it is beneficial to have both $\ell_{\bm\theta}$ and $\lambda_{\bm\theta}$ parametrised to increase the number of degrees of freedom of the approximation scheme \cite{gros_2020_datadriven}, in what follows we propose to only focus on $V_{\bm\theta}^\text{f}$ for computational relief and control performance.

\subsection{Cost-to-go Approximation} \label{sections:background:cost-to-go-approx}

A proper choice of terminal cost $V_{\bm\theta}^\text{f}$ in \eqref{eq:background:mpc:obj} is essential in capturing the cost-to-go for the terminal state $\bm{x}_N$. In general, analytical forms of the true cost-to-go are often unavailable, and approximations must be used instead. As discussed in \Cref{sections:introduction}, the control literature offers various solutions to this challenge. In particular, in this work, we highlight the following approaches from the literature. 

\subsubsection{Nonconvex Case}  \label{sections:background:cost-to-go:psdnn}

for a nonconvex system and stage cost, the cost-to-go approximation can be parametrised as in \cite{abdufattokhov_2024_learning}:
\begin{align}
    \bm{P}_{\bm\theta}(\bm{c}) ={}& L_{\bm\theta}(\bm{c}) L_{\bm\theta}^\top(\bm{c}), \label{eq:backgroun:cholesky-form}  \\
    V_{\bm\theta}^\text{f,psd}(\bm{x}, \bm{c}) ={}& 
    \norm*{ \bm{x} - \bm{x}^\text{f}_{\bm\theta}(\bm{c}) }_{ \bm{P}_{\bm\theta}(\bm{c}) }^2,
\end{align}
where $\bm{c} \in \mathbb{R}^{n_c}$ is the task-relevant context available at the current time step (which can include any information, e.g., state $\bm{x}$, previous actions, references, etc.), and both $\funcdef{\bm{x}^\text{f}_{\bm\theta}}{\mathbb{R}^{n_c}}{\mathbb{R}^{n_s}}$ and $\funcdef{L_{\bm\theta}}{\mathbb{R}^{n_c}}{\mathbb{R}^{n_s \times n_s}}$ are represented by two neural networks (NNs), whose parameters are meant as included in $\bm\theta$. In particular, $L_{\bm\theta}(\bm{c})$ from \eqref{eq:backgroun:cholesky-form} is a lower triangular matrix with only $\frac{1}{2} n_s (n_s + 1)$ free entries. This Cholesky decomposition-like form allows the approximate terminal cost to be positive semidefinite (PSD) w.r.t. $\bm{x}$ by construction. For a fixed $\bm{c}$, this makes optimising over the ensuing quadratic form relatively easy and cheap. At the same time, the quadratic form is context-dependent, meaning its value and gradient information will change from time step to time step, making the approximation also time-dependent. Lastly, the approach is quite malleable as the two NNs can be seamlessly scaled down or up as needed. 

\subsubsection{Convex Case}  \label{sections:background:cost-to-go:pwqnn}

in the case of constrained linear time-invariant systems with quadratic regulation cost, it is well-known that the optimal value function is convex piecewise quadratic (PWQ) \cite{bemporad_2002_explicit}. This result can also be extended to the stochastic setting with zero-mean, time-uncorrelated Gaussian disturbances \cite{lim_1996_separation}. In such a case where the value function is known to have (exactly or even approximately) a convex PWQ shape, \cite{he_2024_apprxoimate} suggests the use of the approximation
\begin{align}
    \bm\varphi(\bm{x}) ={}& \text{ReLU}\left( 
        \bm{W}_{\bm\theta} \bm{x} + \bm{b}_{\bm\theta} 
    \right),  \\
    V_{\bm\theta}^\text{f,pwq}(\bm{x}) ={}& 
    \bm{w}_{\bm\theta}^\top 
    \bigl( \bm\varphi(\bm{x}) \odot \bm\varphi(\bm{x}) \bigr),
\end{align}
where $\bm{W}_{\bm\theta} \in \mathbb{R}^{m \times n_s}$, $\bm{b}_{\bm\theta} \in \mathbb{R}_{<0}^m$ and $\bm{w}_{\bm\theta} \in \mathbb{R}_{\geq 0}^m$ are the adjustable weights and biases of the NN, and $\bm\varphi \in \mathbb{R}^m$ its hidden features. By enforcing $\bm{b}_{\bm\theta} < 0$ and $\bm{w}_{\bm\theta} \geq 0$, it is shown in \cite{he_2024_apprxoimate} that this function is PWQ and convex w.r.t. $\bm{x}$. The advantage of this approximation lies in its scalability (by appropriately selecting the hidden size $m$) and ability to represent any PWQ convex functions by construction.

\section{METHODOLOGY} \label{sections:methodology}

This section introduces a sample-based approximation to the stochastic MPC problem. Similarly to, e.g., \cite{abdufattokhov_2024_learning,he_2024_apprxoimate}, we propose to employ a learnable terminal cost formulation, coupled with a short prediction horizon, to mitigate the computational complexity induced by the sampling approach. At the same time, to preserve the probabilistic safety of the closed-loop state trajectories despite the increased myopia of the controller, we leverage the notion of CBF to guarantee step-wise forward invariance of the safe set with high probability. We adopt RL to perform training of both the terminal cost and the CBF class $\mathcal{K}$ function in an end-to-end fashion. 

\subsection{Probabilistic Control Barrier Function}

In this work, we propose to leverage the CBF framework to guarantee safety. Again, it is essential to remark that, due to the stochasticity affecting the system, in general safety cannot be guaranteed with unit probability. Rather, we will take a probabilistic approach.

\begin{definition}[$N$-Step $\varepsilon$-Control Invariant Set \cite{gao_2021_computing}]
    A set $\mathcal{Q} \subseteq \mathbb{R}^{n_s}$ is $N$-step $\varepsilon$-control invariant for system \eqref{eq:background:dynamics} if, for any $\bm{s}_t \in \mathcal{Q}$, there exists a control policy such that
    \begin{equation}
        \prob*{ 
            \bigcap_{\tau=1}^N \bm{s}_{t+\tau} \in \mathcal{Q} 
        } \geq 1 - \varepsilon.
    \end{equation}
\end{definition}

\begin{definition}[Probabilistic Control Barrier Function \cite{wang_2024_stochastic}]
    For system \eqref{eq:background:dynamics} and safe set $\mathcal{C} \subseteq \mathcal{S}$, the continuous function $\funcdef{h}{\mathcal{S}}{\mathbb{R}}$ is a discrete-time probabilistic CBF if there exist a class $\mathcal{K}$ function $\funcdef{\alpha}{[0,a)}{[0,\infty)}$, $\alpha(y) \leq y$, $\forall y \geq 0$, and a control action $\bm{a}_t \in \mathcal{A}$ such that, with $\xi \in [0, 1)$, it holds that
    \begin{equation}  \label{eq:methodology:probabilistic-cbf}
        \prob*{
            h(\bm{s}_{t+1}) - h(\bm{s}_t) \geq - \alpha \bigl( h(\bm{s}_t) \bigr)
        } \geq 1 - \xi,
        \quad \forall t \in \mathbb{N}.
    \end{equation}
\end{definition}
Note that the above follows straightforwardly from the standard CBF definition, on top of which the probability operator $\mathbb{P}$ has been applied since the state is now a random variable. Now, we can state a result on the step-wise probabilistic invariance guarantee for the set $\mathcal{C}$ thanks to the CBF condition.

\begin{theorem}
    Given a safe set $\mathcal{C} \subseteq \mathcal{S}$ defined by the continuous function $\funcdef{h}{\mathcal{S}}{\mathbb{R}}$ and current state $\bm{s}_t \in \mathcal{C}$, if $h$ is a discrete-time probabilistic CBF, any control policy satisfying \eqref{eq:methodology:probabilistic-cbf} with $\xi \leq \frac{\varepsilon}{N}$ renders the set $\mathcal{C}$ $N$-step $\varepsilon$-control invariant.
\end{theorem}
\begin{proof}
    The proof is similar to that of \cite[Theorem~2]{wang_2024_stochastic}. By complement, the joint safety condition along an $N$-step trajectory is satisfied as long as 
    \begin{equation}
        \prob*{ 
            \bigcup_{\tau=1}^N \bm{s}_{t+\tau} \notin \mathcal{C}
        } \leq \varepsilon.
    \end{equation}
    Applying the union bound, we get that
    \begin{equation}
        \prob*{ 
            \bigcup_{\tau=1}^N \bm{s}_{t+\tau} \notin \mathcal{C}
        } \leq
        \sum_{\tau=1}^N \prob*{ \bm{s}_{t+\tau} \notin \mathcal{C} }.
    \end{equation}
    To ensure the joint probability of violation is at most $\varepsilon$, it is thus sufficient to require $\sum_{\tau=1}^N \prob*{ \bm{s}_{t+\tau} \notin \mathcal{C} } \leq \varepsilon$. The simplest choice is to allocate the risk uniformly per time step, i.e., we need to satisfy
    \begin{equation}
        \prob*{ 
            \bm{s}_{t+\tau} \in \mathcal{C}
        } \geq 1 - \frac{\varepsilon}{N}, 
        \ \tau=1,\ldots,N.
    \end{equation}
    To achieve this, we select $\xi \leq \frac{\varepsilon}{N}$ and compute the action at each time step $t+\tau$ according to \eqref{eq:methodology:probabilistic-cbf}, so that 
    \begin{multline}
        \prob*{ h(\bm{s}_{t+\tau+1}) \geq 0}
        \\ 
        \geq \prob*{
            h(\bm{s}_{t+\tau+1}) \geq h(\bm{s}_{t+\tau}) - \alpha \bigl( h(\bm{s}_{t+\tau}) \bigr)
        } 
        \\ \geq 1 - \xi 
        \geq 1 - \frac{\varepsilon}{N},
        \ \tau=1,\ldots,N - 1,
    \end{multline}
    where the first inequality leverages the fact that $\bm{s}_t \in \mathcal{C} \Rightarrow h(\bm{s}_t) \geq 0$, and the property $\alpha(y) \leq y$, $\forall y \geq 0$.
\end{proof}
This result implies that, if the control policy acts accordingly to \eqref{eq:methodology:probabilistic-cbf} with $\xi$ properly selected as $\frac{\varepsilon}{N}$, the state trajectory can occasionally leave the safe set $\mathcal{C}$ but the chance of doing so is bounded by $\varepsilon$. Note that, while leveraging the CBF condition is beneficial to safety, there are still some open issues. In particular, finding a control input directly via \eqref{eq:methodology:probabilistic-cbf} is in general challenging due to the probability operator: a distributional characterisation of its argument may be challenging due to the possible nonlinear nature of $f$, $h$, and/or $\alpha$, and would also require exact knowledge of the distribution $\mathcal{W}$. Additionally, it is well-known that properties of the ensuing control policy (such as performance) are dependent on the selection of a proper class $\mathcal{K}$ function. 

\subsection{Sample-based MPC Approximation}

We can now introduce the proposed sample-based approximation of \eqref{eq:background:mpc}. Let us introduce a shortened horizon $\overline{N} \ll N$. At time step $t$, assume $M$ samples $\{ \bm\omega_\tau^{(i)} \}_{\tau=t}^{t + \overline{N} - 1}$, $i=1,\ldots,M$, are available (see \Cref{assumptions:background:iid-samples}).\footnote{Since disturbances could be time-correlated, these samples must be drawn by sampling whole sequences from $\mathcal{W}$ and then considering only the first $\overline{N}-1$ elements.} Then, we can replace the original intractable formulation \eqref{eq:background:mpc} with the following scheme:
\begin{mini!}
    { \scriptstyle \left\{ \bm{u}_k \right\}_{k=0}^{\overline{N}-1} }{ 
        \lambda_{\bm\theta}\left(\bm{s}_t\right)
        \nonumber
    }{ \label{eq:methodology:scmpc} }{}
    \breakObjective{
        + \frac{1}{M} \sum_{i=1}^M \Biggl[ 
            \sum_{k=0}^{\overline{N}-1} \ell_{\bm\theta}\bigl(\bm{x}_k^{(i)}, \bm{u}_k\bigr)
            + V^\text{f}_{\bm\theta}\bigl(\bm{x}_{\overline{N}}^{(i)}\bigr)
        \Biggr]
    }
    \addConstraint{ 
        \bm{u}_k \in \mathcal{A}, 
        \hspace{59pt} k=0,\ldots,\overline{N}-1, 
    }
    \addConstraint{ 
        \bm{x}_0^{(i)} = \bm{s}_t, 
        \hspace{52pt} i=1,\ldots,M, 
    }
    \addConstraint{ 
        \bm{x}_{k+1}^{(i)} = f\bigl(\bm{x}_k^{(i)}, \bm{u}_k, \bm\omega_k^{(i)}\bigr), 
        \nonumber
    }
    \addConstraint{ \hspace{33pt} i=1,\ldots,M, \ k=0,\ldots,\overline{N}-1, }
    \addConstraint{ 
        h\bigl(\bm{x}_{k+1}^{(i)}\bigr) - h\bigl(\bm{x}_k^{(i)}\bigr) \geq \zeta - \alpha_{\bm\theta}\Bigl( h\bigl(\bm{x}_k^{(i)}\bigr) \Bigr), 
        \nonumber 
    }
    \addConstraint{ \hspace{33pt} i=1,\ldots,M, \ k=0,\ldots,\overline{N}-1. \label{eq:methodology:scmpc:constr:safety} }
\end{mini!}
Major differences lie in the safety condition \eqref{eq:background:mpc:constr:last} being replaced by the proposed probabilistic CBF formulation \eqref{eq:methodology:scmpc:constr:safety}, and the probabilistic operators, e.g., the expectation in \eqref{eq:background:mpc:obj}, by sample approximation. By selecting a (much) shorter horizon, we are able to counterbalance the increased size of the optimisation problem. However, myopic policies tend to be less safe. For this reason, we leverage the CBF to ensure control invariance. Still, to avoid jeopardising safety due to the sample-based approximation, we stress that the number of samples $M$ must be selected in such a way to guarantee that the CBF condition is satisfied with probability $1-\frac{\varepsilon}{N}$. In what follows, we discuss how to select $M$ in order to achieve this, thus preserving safety with confidence $\varepsilon$. Note that $\zeta \geq 0$ in \eqref{eq:methodology:scmpc:constr:safety} is a (usually small) scalar required to ensure the probabilistic guarantees discussed below. It is trivial to check that, being nonnegative, its presence does not jeopardise the CBF validity. If the problem \eqref{eq:methodology:scmpc} is convex, $\zeta$ can be freely set to zero; for the generic nonconvex case, $\zeta \neq 0$ (see proof of \Cref{theorems:safety}).

\begin{assumption}[Recursive Feasibility]  \label{assumptions:background:recursive-feas}
    Under the ensuing control policy, the sampled-based MPC scheme \eqref{eq:methodology:scmpc} admits a feasible solution at every time step $t \in \mathbb{N}$ almost surely.
\end{assumption}

This assumption is a requirement for the following result, and is standard in other works, e.g., \cite{schildbach_2014_scenario,campi_2019_scenario}. At first, it might appear restrictive but in practice hard constraints are often replaced by soft constraints in stochastic/learning settings. This choice is corroborated by the probabilistic nature of the control problem, i.e., violations cannot be avoided with unit probability (without further assumptions). Furthermore, this choice is also helpful in the context of RL: during learning, it is beneficial for the RL agent to violate constraints occasionally and receive appropriate penalties so as to learn better to discern safe and unsafe behaviours \cite{gros_2020_datadriven}.

\begin{theorem}  \label{theorems:safety}
    Given a confidence parameter $\beta \in (0,1)$, there exists a minimum number of samples $M$ for which the solution $\left\{ \bm{u}_k^\star \right\}_{k=0}^{\overline{N}-1}$ to the sample-based optimisation problem \eqref{eq:methodology:scmpc} satisfies
    \begin{multline}  \label{eq:theorem:safety}
        \prob*{
            \bigcap_{k=0}^{\overline{N}-1} 
            h\left(\bm{x}_{k+1}^\star\right) - h\left(\bm{x}_k^\star\right) \geq \zeta - \alpha\bigl( h\left(\bm{x}_k^\star\right) \bigr)
        } 
        \ge 1 - \frac{\varepsilon}{N}
    \end{multline}
    with probability no smaller than $1 - \beta$, where $\bm{x}_0^\star=\bm{s}_t$ and $\bm{x}_{k+1}^\star = f\bigl(\bm{x}_k^\star, \bm{u}_k^\star, \bm\omega_k\bigr)$.
\end{theorem}
\begin{proof}
    For sake of brevity, for any feasible solution to \eqref{eq:methodology:scmpc} we drop the implicit dependency of $\left\{ \bm{x}_k \right\}_{k=0}^{\overline{N}}$ on $\left\{ \bm{u}_k \right\}_{k=0}^{\overline{N}-1}$. Given $\bm{s}_t$, define the violation probability of a solution as
    \begin{multline}
        V_{\bm{u}} = \prob*{
            \bigcup_{k=0}^{\overline{N}-1} 
            h\left(\bm{x}_{k+1}\right) - h\left(\bm{x}_k\right) < \zeta - \alpha\bigl( h\left(\bm{x}_k\right) \bigr)
        }.
    \end{multline}
    Take $\xi = \frac{\varepsilon}{N}$. Analogously to \Cref{sections:background:cost-to-go-approx}, we distinguish between two cases.

    In case \eqref{eq:methodology:scmpc} is nonconvex, let $d_\mathcal{A} = \sup_{\bm{a},\bm{a}' \in \mathcal{A}} \norm{ \bm{a} - \bm{a}' }_\infty$ be the diameter of $\mathcal{A}$. Because $\alpha(y) \leq y$, $\forall y \geq 0$, and $\alpha$ is strictly increasing, $\alpha$ is Lipschitz continuous with constant 1. Given the Lipschitz constants of $f$ and $h$, the CBF constraint \eqref{eq:methodology:scmpc:constr:safety} is also Lipschitz continuous with constant at most $L_\text{CBF} = L_h L_f + L_h + L_h$. By \cite[Theorem~10]{luedtke_2008_sample} we have
    \begin{equation}
        \prob*{
            V_{\bm{u}^\star} > \xi
        }
        \leq 
        \left\lceil \frac{2}{\xi} \right\rceil
        \left\lceil \frac{ 2 d_\mathcal{A} L_\text{CBF} }{\zeta} \right\rceil^{\overline{N} n_a} 
        e^{ -\frac{1}{2} M \xi^2 }.
    \end{equation}
    By requiring that the right-hand side be $\leq \beta$, we get
    \begin{equation}  \label{eq:methodology:num-samples-nonconvex}
        M \geq \frac{2}{\xi^2} \left( 
            \ln{\beta^{-1}} 
            + \overline{N} n_a \ln{ \left\lceil \frac{ 2 d_\mathcal{A} L_\text{CBF} }{\zeta} \right\rceil }
            + \ln{ \left\lceil \frac{2}{\xi} \right\rceil }
        \right).
    \end{equation}
    
    Let us tackle the special case in which \eqref{eq:methodology:scmpc} is convex w.r.t. its decision variables (it needs not be convex w.r.t. the disturbance). Contrarily to the previous case, here we can consider $\zeta=0$ as it is not required. The scenario approach theory \cite{schildbach_2014_scenario,campi_2019_scenario,campi_2008_exact} shows that the probability of violation at the optimal solution of \eqref{eq:methodology:scmpc} is (possibly tightly) bounded by \cite[Theorem~1]{campi_2008_exact}
    \begin{equation}
        \prob*{
            V_{\bm{u}^\star} > \xi
        }
        \leq \sum_{j=0}^{\overline{N} n_a - 1} 
        \begin{pmatrix} M \\ j \end{pmatrix} 
        \xi^j \left( 1 - \xi \right)^{M - j}.
    \end{equation}
    By requiring that the right-hand side be $\leq \beta$, we obtain
    \begin{equation}  \label{eq:methodology:num-samples-convex}
        M \geq \frac{2}{\xi} \left( \ln{\beta^{-1}} + \overline{N} n_a \right).
    \end{equation}
\end{proof}

Despite of arguably limited applicability, this theorem importantly confirms the intuition that, as the sample size $M$ increases, the confidence at which the safety condition is satisfied increases. Also, note that, while the horizon $\overline{N}$ has been shrunk to combat the computational complexity due to the sampling scheme, the probabilistic safety condition has been left untouched and is still imposed over the original $N$-step trajectory (see right-hand side of \eqref{eq:theorem:safety} where the risk of violation is allocated over $N$ steps instead of $\overline{N}$).

\subsection{RL Algorithm}

Note that most of the major components in \eqref{eq:methodology:scmpc} are parametrised in $\bm\theta$, including the class $\mathcal{K}$ function $\alpha_{\bm\theta}$ and the terminal cost function $V^\text{f}_{\bm\theta}$. We propose to adjust this parametrisation in closed loop via an MPC-based RL algorithm \cite{gros_2020_datadriven}. This approach solves the original safe RL problem \eqref{eq:background:safe-rl} as the safety constraint is taken into account into the MPC function approximation while the performance cost \eqref{eq:background:safe-rl:performance} is minimised by a gradient-based RL method. Among the advantages of this approach is the fact that it bypasses the need to manually craft and select the parametrised components, which are instead adjusted by RL via interactions with the environment. This encompasses the ability also to learn $\alpha_{\bm\theta}$, yielding an intrinsically adaptive CBF that can automatically balance the trade-off between trajectory safety and control performance.

Because we explicitly include a learnable terminal cost term in the objective, a value-based method is leveraged here. In particular, we propose the use of Q-learning \cite{watkins_1989_learning}. Briefly, Q-learning indirectly finds the optimal policy by solving $\min_{\bm\theta} \expval*{ \norm*{\ell(\bm{s}, \bm{a}) + V_{\bm\theta}(\bm{s}_+) - Q_{\bm\theta}(\bm{s},\bm{a}) }^2 }$, where $V_{\bm\theta}$ and $Q_{\bm\theta}$ are defined in \eqref{eq:background:V} and \eqref{eq:background:Q}. The problem can be minimised via, e.g., gradient descent updates
\begin{equation}
    \bm\theta \leftarrow \bm\theta + \eta \delta \nabla_{\bm\theta} Q_{\bm\theta}(\bm{s},\bm{a}),
\end{equation}
with $\eta > 0$ a properly selected learning rate and $\delta$ the temporal difference error. For the computation of $\nabla_{\bm\theta} Q_{\bm\theta}$, while not straightforward, nonlinear sensitivity analysis of the MPC scheme \eqref{eq:methodology:scmpc} shows that this sensitivity coincides with the partial derivative of the Lagrangian w.r.t. $\bm\theta$ at the optimal primal-dual solution \cite{buskens_2001_sensitivity}. Details on the implementation can be found in, e.g., \cite{zanon_2021_safe,airaldi_2023_learning}.

\section{NUMERICAL EXPERIMENT}  \label{sections:numerical-exp}

In this section, we test the proposed methodology on a numerical case. The experiment was implemented in Python 3.12.6 and conducted on a server with 16 AMD EPYC 7252 (3.1 GHz) processors and 252GB RAM. The optimisation problems were formulated with CasADi \cite{andersson_2019_casadi}, and solved via Gurobi \cite{gurobi}. Source code and results are available in the following repository: \url{https://github.com/FilippoAiraldi/mpcrl-cbf}.

\subsection{Problem Description}

Consider the stochastic LTI system
$f(\bm{s}_t,\bm{a}_t,\omega_t) = \bm{A} \bm{s}_t + \bm{B} \bm{a}_t + \bm{E} \omega_t$ with
\begin{equation}
    \bm{A} = \begin{bmatrix} 1 & 0.4 \\ -0.1 & 1 \end{bmatrix}, \
    \bm{B} = \begin{bmatrix} 1 & 0.05 \\ 0.5 & 1 \end{bmatrix}, \
    \bm{E} = \begin{bmatrix} 0.03 \\ 0.01 \end{bmatrix},
\end{equation}
where the disturbances are time-uncorrelated zero-mean normally distributed, i.e., $\expval*{\omega_t} = 0$ and $\expval*{\omega_i \omega_j} \propto \delta(i-j)$. The control space is $\mathcal{A} = \left\{ \bm{a} \in \mathbb{R}^2 : \norm{\bm{a}}_\infty \leq 0.5 \right\}$. The safe set is defined as $\mathcal{C} = \left\{ \bm{s} \in \mathbb{R}^2 : \norm{\bm{s}}_\infty \leq 3 \right\}$, where the infinity-norm state constraint is turned into four separate CBFs $\funcdef{h_j}{\mathbb{R}^2}{\mathbb{R}}$, $j=1,\ldots,4$, i.e., $h_{2i-1}(\bm{s}) = 3 - s_i$ and $h_{2i}(\bm{s}) = 3 + s_i$, $i=1,2$. The RL stage cost includes quadratic terms alongside penalties for the violation of the safety condition:
\begin{equation}  \label{eq:numerical-exp:lti:stage-cost}
    \ell(\bm{s},\bm{a}) = 
    \norm{ \bm{s} }_{\bm{Q}}^2
    + \norm{ \bm{u} }_{\bm{R}}^2
    - c \sum_{j=1}^4 \min \left\{ 0, h_j(\bm{s}) \right\},
\end{equation}
with $\bm{Q} = \bm{I}_{2\times2}$, $\bm{R} = 0.1 \bm{I}_{2\times2}$, and $c = 10^3$. The length of a single episode is set to $T=30$ time steps. 

\subsection{MPC and RL Implementation}

Given the current state $\bm{s}_t$, the following unit-horizon MPC scheme derived from \eqref{eq:background:mpc} and \eqref{eq:methodology:scmpc} is employed as function approximation with $M = 32$ samples\footnote{Although not selected according to \eqref{eq:methodology:num-samples-convex} (which in principle provides a conservative bound on the number of samples required), this $M$ already leads to a sufficiently safe and computationally not-too-expensive policy in our test environment.}:
\begin{mini!}
    { \scriptstyle \bm{u}_0, \Sigma }{ 
        \ell( \bm{s}_t, \bm{u}_0 )
        + \frac{1}{M} \sum_{i=1}^M \left[
            c \sum_{j=1}^4 \sigma_j^{(i)}
            + V_{\bm\theta}^\text{f,pwq}\bigl(\bm{x}_{1}^{(i)}  \bigr)
        \right] 
        \label{eq:numerical-exp:lti:scmpc:obj}
    }{ \label{eq:numerical-exp:lti:scmpc} }{ }
    \addConstraint{ -0.5 \leq \bm{u}_0 \leq 0.5, }
    \addConstraint{ 
        \bm{x}_{1}^{(i)} 
        = f\bigl(\bm{s}_t, \bm{u}_0, \omega_0^{(i)}\bigr),
        \hspace{27pt} i=1,\ldots,M,
    }
    \addConstraint{ 
        h_j\bigl( \bm{x}_1^{(i)} \bigr) 
        - \left( 1 - \gamma_{\bm\theta,j} \right) h_j(\bm{s}_t)
        + \sigma_j^{(i)}
        \ge 0,
        \nonumber 
    }
    \addConstraint{ 
        \hspace{63pt} j=1,\dots,4, \ i=1,\ldots,M, 
        \label{eq:numerical-exp:lti:scmpc:con:cbf}
    }
    \addConstraint{ \sigma_j^{(i)} \geq 0, \hspace{24.75pt} j=1,\dots,4, \ i=1,\ldots,M. }
\end{mini!}
As terminal cost approximation in \eqref{eq:numerical-exp:lti:scmpc:obj}, the convex PWQ function $V_{\bm\theta}^\text{f,pwq}$ is employed with a hidden size of 16 neurons. For each CBF constraint \eqref{eq:numerical-exp:lti:scmpc:con:cbf}, the corresponding class $\mathcal{K}$ function is parametrised linearly, i.e., $\alpha_{\bm\theta,j}(y) = \gamma_{\bm\theta,j} y$, $j=1,\ldots,4$, where $\gamma_{\bm\theta,j} \in [0, 1]$ is an adjustable scalar value. The whole MPC learnable parametrisation is therefore
\begin{equation}
    \bm\theta = \left( \bm{W}_{\bm\theta}, \bm{b}_{\bm\theta}, \bm{w}_{\bm\theta}, \gamma_{\bm\theta,1}, \dots, \gamma_{\bm\theta,4} \right),
\end{equation}
where $\bm{W}_{\bm\theta}$, $ \bm{b}_{\bm\theta}$ and $ \bm{w}_{\bm\theta}$ are defined per \Cref{sections:background:cost-to-go:pwqnn}. Note that the CBF constraints \eqref{eq:numerical-exp:lti:scmpc:con:cbf} have been relaxed via slack variables $\Sigma = \bigl\{ \sigma_j^{(i)}, i=1,\ldots,M, j=1,\ldots,4 \bigr\}$ to preserve feasibility (see \Cref{assumptions:background:recursive-feas}), and, since the problem is convex, we set $\zeta = 0$.

The MPC parametrisation is initialised to uniformly random values for the PWQ NN, and to $\gamma_{\bm\theta,j} = 0.7$ for all the CBF parameters. A Q-learning agent is trained for 1000 episodes with a learning rate of $0.005$ via \textit{rmsprop} \cite{tieleman_2012_rmsprop}. The parametrisation is updated at the end of each episode, based on the experiences observed in the last episode. Since $\gamma_{\bm\theta,j}$ and part of the parametrisation of the PWQ NN must be constrained, a constrained step update of $\bm\theta$ is performed \cite{airaldi_2023_learning}. To induce exploration in an epsilon-greedy fashion, a term $\bm{q}^\top \bm{u}_0$ is added to the objective \eqref{eq:numerical-exp:lti:scmpc:obj}, where $\bm{q} \sim \mathcal{N}\left(\bm{0}_{2 \times 2}, \rho_{\bm{q}} \bm{I}_{2 \times 2}\right)$. The exploration scale $\rho_{\bm{q}}$ and probability both start at 1 but decay by a factor of 0.997 after each episode. Lastly, the training procedure is repeated for 10 differently randomly seeded agents to account for randomness.

\subsection{Results}

Numerical results corroborate the capability of the proposed framework in appropriately learning the terminal cost MPC component via RL, as well as the effectiveness of the learned policy compared to a full-length horizon MPC controller. \Cref{fig:lti:nrmse-and-r2} shows the evolution of the PWQ approximation w.r.t. the explicit optimal solution, computed in accordance to \cite{lim_1996_separation}. Convergence of both the normalised error and the $R^2$ coefficient during training provides empirical evidence that the Q-learning algorithm is able to steer the $V_{\bm\theta}^\text{f}$ term towards the real optimal one. 
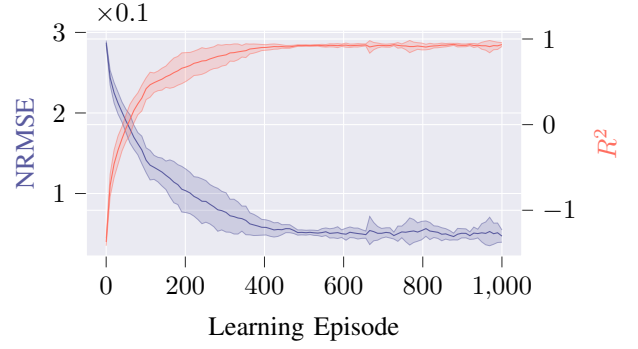
\begin{figure}
    \centering
    \begin{tikzpicture}
    \pgfplotstableread{data/lti/mpcrl_nrmse_and_r2.dat}\data

    \begin{axis}[
        name=plot-nrmse,
        xlabel=Learning Episode,
        ylabel={\textcolor{C0}{NRMSE}},
        height=\axisdefaultheight,
        width=0.85\axisdefaultwidth,
        scaled y ticks=real:0.1,
    ]
        \addplot [C0] table [x=episode, y=nrmse-avg] {\data};
        \addplot [C0, ultra thin, opacity=0.5, name path=nrmse-lb] table [x=episode, y expr={\thisrow{nrmse-avg}-\thisrow{nrmse-std}}] {\data};
        \addplot [C0, ultra thin, opacity=0.5, name path=nrmse-ub] table [x=episode, y expr={\thisrow{nrmse-avg}+\thisrow{nrmse-std}}] {\data};
        \addplot [C0, opacity=0.2] fill between [of=nrmse-lb and nrmse-ub, reverse=true];
    \end{axis}

    \begin{axis}[
        name=plot-r2,
        axis x line=none,
        axis y line*=right,
        axis background/.style={draw=none},
        axis line style={draw=none},
        ylabel={\textcolor{C1}{$R^2$}},
        height=\axisdefaultheight,
        width=0.85\axisdefaultwidth,
    ]
        \addplot [C1] table [x=episode, y=r2-avg] {\data};
        \addplot [C1, ultra thin, opacity=0.5, name path=r2-lb] table [x=episode, y expr={\thisrow{r2-avg}-\thisrow{r2-std}}] {\data};
        \addplot [C1, ultra thin, opacity=0.5, name path=r2-ub] table [x=episode, y expr={\thisrow{r2-avg}+\thisrow{r2-std}}] {\data};
        \addplot [C1, opacity=0.2] fill between [of=r2-lb and r2-ub, reverse=true];
    \end{axis}
\end{tikzpicture}
    \caption{Evolution of the learned terminal cost approximation in terms of normalised RMSE and coefficient of determination w.r.t. the optimal cost-to-go function for the constrained stochastic LTI experiment. Average results $\pm$ one standard deviation over 10 different seeds are reported.}
    \label{fig:lti:nrmse-and-r2}
\end{figure}
\Cref{fig:lti:gamma-evolution} reports the evolution of the CBF parameters $\gamma_{\bm\theta,1}$ and $\gamma_{\bm\theta,3}$, which correspond to the lower and upper bounds on the first state. These are of more interest because, due to the dynamics, most of the violations tend to occur in these two constraints (the other two parameters $\gamma_{\bm\theta,2}$ and $\gamma_{\bm\theta,4}$ are omitted as they do not change as much during learning). It is important to stress again that these CBF parameters (as well as all other parameters included in $\bm\theta$) are adjusted by Q-learning to enhance closed-loop performance. Because constraint violations are included in the cost \eqref{eq:numerical-exp:lti:stage-cost} as penalty term, safety is only indirectly taken into account by the RL algorithm. Nonetheless, since the parameters $\gamma_{\bm\theta,j}$ are constrained to the interval $[0,1]$ in each update, the CBFs remain valid throughout the learning process. As a matter of fact, during training, the MPC-based RL policy achieves a small empirical probability ($0.0942 \pm 0.00483 \%$) of violating any constraint. 
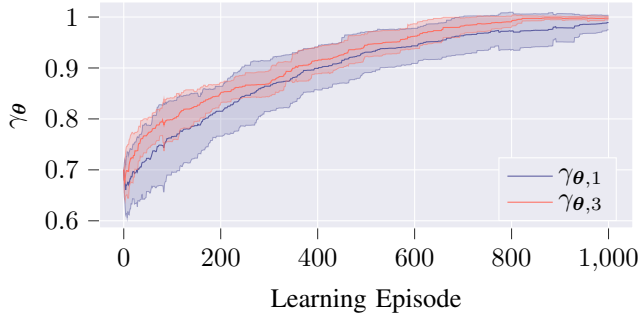
\begin{figure}
    \centering
    \begin{tikzpicture}
    \pgfplotstableread{data/lti/mpcrl_gamma.dat}\data

    \begin{axis}[
        name=plot-gamma,
        xlabel=Learning Episode,
        ylabel=$\gamma_{\bm\theta}$,
        legend pos=south east,
    ]
        \addplot [C0] table [x=episode, y=gamma0-avg] {\data};
        \addplot [
            C0, ultra thin, opacity=0.5, name path=gamma0-lb
        ] table [x=episode, y expr={\thisrow{gamma0-avg}-\thisrow{gamma0-std}}] {\data};
        \addplot [
            C0, ultra thin, opacity=0.5, name path=gamma0-ub
        ] table [x=episode, y expr={\thisrow{gamma0-avg}+\thisrow{gamma0-std}}] {\data};
        \addplot [C0, opacity=0.2] fill between [of=gamma0-lb and gamma0-ub, reverse=true];

        \addplot [C1] table [x=episode, y=gamma2-avg] {\data};
        \addplot [
            C1, ultra thin, opacity=0.5, name path=gamma2-lb
        ] table [x=episode, y expr={\thisrow{gamma2-avg}-\thisrow{gamma2-std}}] {\data};
        \addplot [
            C1, ultra thin, opacity=0.5, name path=gamma2-ub
        ] table [x=episode, y expr={\thisrow{gamma2-avg}+\thisrow{gamma2-std}}] {\data};
        \addplot [C1, opacity=0.2] fill between [of=gamma2-lb and gamma2-ub, reverse=true];

        \legend{$\gamma_{\bm\theta,1}$,,,,$\gamma_{\bm\theta,3}$}
    \end{axis}
\end{tikzpicture}
    \caption{Evolution of two of the linear class $\mathcal{K}$ function learnable coefficients. Average results $\pm$ one standard deviation over 10 different seeds are reported.}
    \label{fig:lti:gamma-evolution}
\end{figure}

After the training phase, the learned MPC-based RL policy is evaluated against a full-length horizon stochastic MPC policy. The latter is similar to \eqref{eq:numerical-exp:lti:scmpc} but is fixed (i.e., it contains no learnable terms) and has a horizon of 12 (instead of 1), which was found to be sufficient to achieve the lowest closed-loop cost (see, e.g., \cite{chmielewski_1996_constrained}, for a more thorough discussion on how to find such a horizon). The other hyperparameters, e.g., the number of samples $M$, are the same in both policies. For each evaluation episode, the initial conditions are drawn from the boundary of the maximal invariant set. \Cref{fig:lti:returns-vs-solvertimes} shows the outcomes of this evaluation comparison. Unsurprisingly, CPU time spent online in solving the MPC-RL policy is almost two orders of magnitude shorter than that for the fixed MPC controller, thanks to the corresponding optimisation problem being considerably smaller. However, from the point of view of costs, both policies achieve remarkably similar closed-loop performance despite the difference in horizon lengths. This finding is further validated in \Cref{fig:lti:state-trajectories}, which reports ten state trajectories that highlight how both control policies behave rather similarly. Moreover, both policies exhibit comparable empirical constraint violation probabilities at evaluation ($0.0839 \pm 0.0138\%$ and $0.1 \pm 0.014 \%$, respectively). These probabilities are also in line with the constraint violation probability recorded during training.
\begin{figure}
    \centering
    \begin{tikzpicture}
    \pgfplotstableread{data/lti/mpc_returns.dat}\MpcReturns
    \pgfplotstableread{data/lti/mpc_returns_quartiles.dat}\MpcReturnsQuartiles
    \pgfplotstableread{data/lti/mpc_solvertimes.dat}\MpcSolverTimes
    \pgfplotstableread{data/lti/mpc_solvertimes_quartiles.dat}\MpcSolverTimesQuartiles
    
    \pgfplotstableread{data/lti/mpcrl_returns.dat}\MpcRlReturns
    \pgfplotstableread{data/lti/mpcrl_returns_quartiles.dat}\MpcRlReturnsQuartiles
    \pgfplotstableread{data/lti/mpcrl_solvertimes.dat}\MpcRlSolverTimes
    \pgfplotstableread{data/lti/mpcrl_solvertimes_quartiles.dat}\MpcRlSolverTimesQuartiles
    
    \def\gap{0.035}  
    \def\distance{2.05 + 2*\gap}  
    \def\marksize{1.0}
    
    \begin{axis}[
        name=plot-returns,
        xmajorgrids=false,
        xmin=-1.1-\gap, xmax=\distance+\gap+1.1,
        xtick={0,\distance},
        xticklabels={MPC, MPC-based RL},
        ylabel={\textcolor{C0}{$J(\pi)$}},
        ytick distance=30,
        height=\axisdefaultheight,
        width=0.83\axisdefaultwidth,
    ]
        \addplot [
            C0, 
            only marks,
            mark=*,
            mark options={draw=white},
            mark size=\marksize,
        ] table [x expr=-\gap, y=returns] {\MpcReturns};
        \addplot [
            C0, 
            ultra thin, 
            opacity=0.5, 
            name path=mpc-returns-center
        ] table [x expr=-\gap, y=kde-point] {\MpcReturns};
        \addplot [
            C0, 
            ultra thin, 
            opacity=0.5, 
            name path=mpc-returns-left
        ] table [x expr=-\gap-\thisrow{kde-pdf}, y=kde-point] {\MpcReturns};
        \addplot [C0, opacity=0.2] fill between [of=mpc-returns-center and mpc-returns-left, reverse=true];
        \pgfplotsinvokeforeach{0,1,2}{
            \pgfplotstablegetelem{#1}{quartile-point}\of\MpcReturnsQuartiles
            \let\x=\pgfplotsretval
            \pgfplotstablegetelem{#1}{quartile-pdf}\of\MpcReturnsQuartiles
            \let\y=\pgfplotsretval
            \ifnum#1=1
                \expanded{\noexpand\draw [C0] (axis cs:-\gap-\y,\x) -- (axis cs:-\gap,\x);}
            \else
                \expanded{\noexpand\draw [C0, densely dashed] (axis cs:-\gap-\y,\x) -- (axis cs:-\gap,\x);}
            \fi
        };

        \addplot [
            C0, 
            only marks,
            mark=*,
            mark options={draw=white},
            mark size=\marksize,
        ] table [x expr=\distance-\gap, y=returns] {\MpcRlReturns};
        \addplot [
            C0, 
            ultra thin, 
            opacity=0.5, 
            name path=mpcrl-returns-center
        ] table [x expr=\distance-\gap, y=kde-point] {\MpcRlReturns};
        \addplot [
            C0, 
            ultra thin, 
            opacity=0.5, 
            name path=mpcrl-returns-left
        ] table [x expr=\distance-\gap-\thisrow{kde-pdf}, y=kde-point] {\MpcRlReturns};
        \addplot [C0, opacity=0.2] fill between [of=mpcrl-returns-center and mpcrl-returns-left, reverse=true];
        \pgfplotsinvokeforeach{0,1,2}{
            \pgfplotstablegetelem{#1}{quartile-point}\of\MpcRlReturnsQuartiles
            \let\x=\pgfplotsretval
            \pgfplotstablegetelem{#1}{quartile-pdf}\of\MpcRlReturnsQuartiles
            \let\y=\pgfplotsretval
            \ifnum#1=1
                \expanded{\noexpand\draw [C0] (axis cs:\distance-\gap-\y,\x) -- (axis cs:\distance-\gap,\x);}
            \else
                \expanded{\noexpand\draw [C0, densely dashed] (axis cs:\distance-\gap-\y,\x) -- (axis cs:\distance-\gap,\x);}
            \fi
        };
    \end{axis}

    \begin{semilogyaxis}[
        name=plot-solvertimes,
        axis background/.style={draw=none},
        axis line style={draw=none},
        axis x line=none,
        xmin=-1.1-\gap, xmax=\distance+\gap+1.1,
        axis y line*=right,
        ytick={1e-4,1e-3,1e-2,1e-1,1e0},
        ylabel={\textcolor{C1}{Avg. CPU Time (s)}},
        ymax=0.6,
        height=\axisdefaultheight,
        width=0.83\axisdefaultwidth,
    ]
        \addplot [
            C1, 
            only marks,
            mark=*,
            mark options={draw=white},
            mark size=\marksize,
        ] table [x expr=\gap, y=solvertimes] {\MpcSolverTimes};
        \addplot [
            C1, 
            ultra thin, 
            opacity=0.5, 
            name path=mpc-solvertimes-center
        ] table [x expr=\gap, y=kde-point] {\MpcSolverTimes};
        \addplot [
            C1, 
            ultra thin, 
            opacity=0.5, 
            name path=mpc-solvertimes-left
        ] table [x expr=\gap+\thisrow{kde-pdf}, y=kde-point] {\MpcSolverTimes};
        \addplot [C1, opacity=0.2] fill between [of=mpc-solvertimes-center and mpc-solvertimes-left, reverse=true];
        \pgfplotsinvokeforeach{0,1,2}{
            \pgfplotstablegetelem{#1}{quartile-point}\of\MpcSolverTimesQuartiles
            \let\x=\pgfplotsretval
            \pgfplotstablegetelem{#1}{quartile-pdf}\of\MpcSolverTimesQuartiles
            \let\y=\pgfplotsretval
            \ifnum#1=1
                \expanded{\noexpand\draw [C1] (axis cs:\gap+\y,\x) -- (axis cs:\gap,\x);}
            \else
                \expanded{\noexpand\draw [C1, densely dashed] (axis cs:\gap+\y,\x) -- (axis cs:\gap,\x);}
            \fi
        };

        \addplot [
            C1, 
            only marks,
            mark=*,
            mark options={draw=white},
            mark size=\marksize,
        ] table [x expr=\distance+\gap, y=solvertimes] {\MpcRlSolverTimes};
        \addplot [
            C1, 
            ultra thin, 
            opacity=0.5, 
            name path=mpcrl-solvertimes-center
        ] table [x expr=\distance+\gap, y=kde-point] {\MpcRlSolverTimes};
        \addplot [
            C1, 
            ultra thin, 
            opacity=0.5, 
            name path=mpcrl-solvertimes-left
        ] table [x expr=\distance+\gap+\thisrow{kde-pdf}, y=kde-point] {\MpcRlSolverTimes};
        \addplot [C1, opacity=0.2] fill between [of=mpcrl-solvertimes-center and mpcrl-solvertimes-left, reverse=true];
        \pgfplotsinvokeforeach{0,1,2}{
            \pgfplotstablegetelem{#1}{quartile-point}\of\MpcRlSolverTimesQuartiles
            \let\x=\pgfplotsretval
            \pgfplotstablegetelem{#1}{quartile-pdf}\of\MpcRlSolverTimesQuartiles
            \let\y=\pgfplotsretval
            \ifnum#1=1
                \expanded{\noexpand\draw [C1] (axis cs:\distance+\gap+\y,\x) -- (axis cs:\distance+\gap,\x);}
            \else
                \expanded{\noexpand\draw [C1, densely dashed] (axis cs:\distance+\gap+\y,\x) -- (axis cs:\distance+\gap,\x);}
            \fi
        };
    \end{semilogyaxis}
\end{tikzpicture}
    \caption{Comparison between the non-learning MPC policy (horizon of 12) and the learned MPC-based RL policy (unit horizon) in terms of the total incurred cost and average solver time over different 1000 episode trials. Lines represent the second (solid) and first and third (dashed) quartiles.}
    \label{fig:lti:returns-vs-solvertimes}
\end{figure}
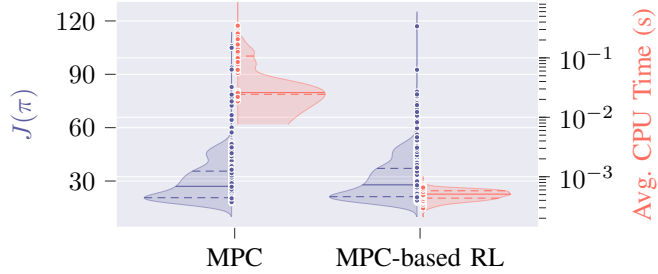
\begin{figure}
    \centering
    \input{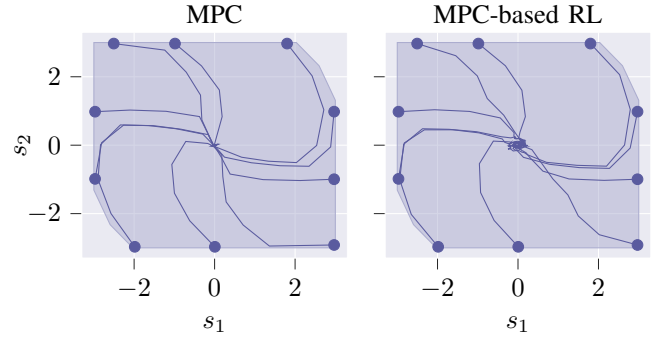}
    \caption{Ten examples of state trajectories recorded during the evaluation of the non-learning MPC policy (horizon of 12) against the learned MPC-based RL policy (unit horizon). Also reported is the maximal invariant set.}
    \label{fig:lti:state-trajectories}
\end{figure}

\section{CONCLUSIONS}  \label{sections:conclusions}

We have proposed a control methodology for stochastic safety-critical systems that merges MPC, CBFs and RL. The parametric MPC controller acts as the backbone, providing the control policy and value function approximation for the RL task. A probabilistic CBF formulation, integrated in the MPC scheme, is put in place to ensure safety of state trajectories with arbitrary probability. To retain tractability of the optimisation problem, the MPC horizon is (substantially) shrunk and a learnable terminal cost is introduced to combat performance drops. RL is then used to adjust the parametrisation of this learnable cost as well as the class $\mathcal{K}$ function, automatically tuning the MPC parametrisation to achieve higher closed-loop performance. A numerical example on a constrained LTI environment showcases the proposed method. Future work will investigate the use of more complex CBF parametrisations (e.g., neural network-based), as well as applications of the proposed methodology to nonlinear systems.

\bibliographystyle{IEEEtran}
\bibliography{references}

\end{document}